\documentclass[a4paper]{article}

\usepackage{balance} 

\usepackage{esvect}
\usepackage{amsfonts}
\usepackage{amsmath}
\usepackage{amssymb}
\usepackage{amsthm}
\usepackage{float}

\usepackage{gensymb}

\newtheorem{theorem}{Theorem}

\newtheorem{corollary}[theorem]{Corollary}
\newtheorem{definition}[theorem]{Definition}

\usepackage{comment}
\usepackage{fullpage}

\usepackage{setspace}

\usepackage{xspace}
\usepackage{tikz}
\newcommand{\reals}{\ensuremath{\mathbb{R}}}

\newcommand{\tdm}{\ensuremath{\textsc{3dm(3)}}\xspace}

\newcommand{\windet}{\ensuremath{\textsc{WinnerDetermination}}\xspace}
\newcommand{\walpri}{\ensuremath{\textsc{WalrasianPricing}}\xspace}

\newcommand{\NP}{\ensuremath{\mathtt{NP}}\xspace}
\newcommand{\p}{\ensuremath{\mathtt{P}}\xspace}
\newcommand{\NC}{\ensuremath{\mathtt{NC}}\xspace}
\newcommand{\RNC}{\ensuremath{\mathtt{RNC}}\xspace}

\newcommand{\problemdef}[3]{
	\begin{center}
		\fbox{ 
			\begin{minipage}{0.95\textwidth}
				\noindent
				\normalsize\textsc{#1}
				
				\vspace{1pt}
				\setlength{\tabcolsep}{3pt}
				\renewcommand{\arraystretch}{1.0}
				\begin{tabularx}{\textwidth}{@{}lX@{}}
					\normalsize\textbf{Input:} 	& \normalsize#2 \\
					\normalsize\textbf{Task:} 	& \normalsize#3 
				\end{tabularx}
			\end{minipage}
		}
	\end{center}
}
\usepackage{tabularx}

\title{Walrasian Equilibria in Markets with Small Demands}

\author{Argyrios Deligkas\thanks{Royal Holloway University of London, UK. 
		Email: \texttt{argyrios.deligkas@rhul.ac.uk}} 
	\and Themistoklis Melissourgos\thanks{Technical University of Munich, Germany. 
		Email: \texttt{themistoklis.melissourgos@tum.de}}
	\and Paul G. Spirakis\thanks{University of Liverpool, UK, and University of Patras, Greece. 
		Email: \texttt{p.spirakis@liverpool.ac.uk}} 
}
\date{\vspace{-1.0cm}}

\begin{document}
	
\maketitle

~\\

\begin{abstract}
	We study the complexity of finding a Walrasian equilibrium in markets where the agents have $k$-demand valuations. These valuations are an extension of unit-demand valuations where a bundle's value is the maximum of its $k$-subsets' values. For unit-demand agents, where the existence of a Walrasian equilibrium is guaranteed, we show that the problem is in quasi-\NC. 
	For $k=2$,  we show that it is \NP-hard to decide if a Walrasian equilibrium exists even if the valuations are fractionally subadditive (\textsc{XOS}), while for $k=3$ the hardness carries over to budget-additive valuations. In addition, we give a polynomial-time algorithm for markets with 2-demand single-minded valuations, or unit-demand valuations. 
\end{abstract}





         
\newcommand{\BibTeX}{\rm B\kern-.05em{\sc i\kern-.025em b}\kern-.08em\TeX}





\setcounter{page}{1}

~\\	
~\\

	\section{Introduction}
One of the most significant problems in market design is finding pricing schemes that guarantee good social welfare under equilibrium. Evidently, the most compelling equilibrium notion in markets with indivisible items is a Walrasian eqilibrium, henceforth WE, \cite{W1896}: an allocation of items to the agents and a pricing, such that every agent maximizes her utility and all items are allocated. By the First Welfare Theorem, WE has the nice property of maximizing social welfare.
The existence of WE seems to heavily rely on the class of valuation functions of the agents. When parameterized by the valuation function class, the existence of WE is (relatively) clear due to Gul and Stracchetti \cite{GS99} and Milgrom \cite{M00}: WE are guaranteed to exist in the class of gross substitutes valuation functions, and beyond this class, there are almost no other valuations where WE are guaranteed. Two of the most central and interesting problems regarding WE are:
\begin{itemize}
	\item[(a)] decide if a WE exists,
	\item[(b)] compute a WE (if it exists).
\end{itemize}

We study the aforementioned two problems when valuation functions are parameterized by 
an integer $k$ which denotes the maximum bundle size $k$ for which every agent is interested. Such a class of {\em $k$-demand} valuation functions can be seen as an extension of the unit-demand functions, where each agent, for a given bundle $X$ values only the most valuable $k$-subset of $X$. 
The main idea behind k-demand valuations is that every agent has some capacity for utilising the items that is either endogenously or exogenously imposed. There are several real-life examples where more than $k$ items have the same value as $k$ of them: a supervisor can effectively supervise up to a limited number of students; a grant investigator can efficiently work up to a limited number of projects; a sports team is allowed to have up to a small number of foreign players (or at least a small number of native players) in the squad; one can hang only a certain number of paintings on their house's walls.
We investigate the complexity of the aforementioned problems when we are restricted to the intersection of the standard valuation classes and the $k$-demand classes. Our results contain hardness results as well as efficient algorithms.

As an example of the effect that $k$-demand valuation functions have on the complexity of these problems, we present {\em unbalanced markets}. In such markets the available items are significantly more than the agents, or vice versa. We provide an algorithm for the aforementioned problems parameterized by $k$. Complemented by a result of Rothkopf et al. \cite{rothkopf1998computationally}, this algorithm concludes that for constant $k$ and appropriate unbalancedness, these problems are in $\p$.

\subsection{Contribution}

In this work we study WE under their classic definition with no relaxation or approximation notions involved. We introduce a hierarchy of valuation functions, parallel to the already existing one. Our valuation functions are called $k$-demand and are a generalization of unit-demand with parameter $k$ that determines at most how many items from a bundle the agent cares about. By definition, it is easy to see that the class of $j$-demand is included in $(j+1)$-demand for any $j \in [m-1]$. The purpose of considering valuation functions from the intersection of some $k$-demand class and some other known class, is to refine the complexity of the WE-related problem. 

Algorithms and hardness results on the existence of WE and/or the problem of computing one in the current literature show an interesting dependence on the parameter $k$ that we define here. For example, existence of WE is guaranteed in the well studied case of unit-demand valuation functions (i.e. $k=1$), and a WE can be computed in polynomial time \cite{DGS86,L17}. Non-existence of WE is established in \cite{RT15} by proving that even computing a welfare maximizing allocation (\windet) is \NP-hard and this is achieved for valuation functions according to which the agents are only interested in at most 2 items (i.e. $k=2$). Furthermore, non-existence of WE and \NP-hardness of \windet is proven for single-minded agents via a reduction to instances where agents are interested in at most 3 items (i.e. $k=3$) \cite{CDS04}. For each of the above cases of $k$ we give improved results: we supplement the ``easy'' case, where $k=1$, with a quasi-\NC algorithm\footnote{This is the first parallel algorithm for computing WE to the authors' knowledge.}, and the ``hard'' cases with stronger \NP-hardness results in the sense that ours imply the existing ones. 

Mixing the standard valuations' hierarchy and the $k$-demand hierarchy results to a two-dimensional landscape of valuation classes that aims to break down the complexity of the WE-related problems. For example, a possible result could be that below some threshold of $k$ and below some standard valuation class, deciding WE existence is in \p. Our results however indicate that this is not the case: even for $k = 2$ and \textsc{XOS} functions \windet is \NP-hard, and therefore deciding existence of WE is also \NP-hard. This is an improvement over the result of Roughgarden and Talgam-Cohen \cite{RT15}, where \NP-hardness is proven for $k=2$ but general functions. Our reduction is entirely different than the one in \cite{RT15}, and in particular, it is from the problem ``3-bounded 3-dimensional matching'' to a market with $n$ agents, $m$ items and $2$-demand \textsc{XOS} valuations. Furthermore, in \cite{LLN06} \windet is proven to be weakly \NP-hard for budget-additive functions by reducing ``knapsack'' to a market with 2 agents, and $m$ items. We show that the problem is strongly \NP-hard for $k$-demand budget-additive functions even for $k = 3$. The case $k=2$ for budget-additive and submodular functions remains open. 

On the positive side, we show a clear dichotomy for the problem of deciding WE existence with single-minded agents. It was proven in \cite{CDS04} that \windet is \NP-hard, via a reduction from ``exact cover by 3-sets'' to a market with single-minded agents who actually used $3$-demand valuations. We show that \windet is solvable in polynomial time for single-minded agents with $2$-demand valuations by a reduction to the maximum weight matching problem. Then, by the decomposition shown at the end of Section \ref{sec:preliminaries}, one can find a WE pricing via an LP (if such a pricing exists).

\subsection{Related Work}

\paragraph{\textbf{Existense of Walrasian Equilibria.}}

The most general class of valuation functions for which existence of WE is guaranteed has been proved by Gul and Stracchetti \cite{GS99} and Milgrom \cite{M00} to be gross substitutes. Other valuation classes (that can be seen as special market settings) outside gross substitutes that guarantee WE existence have also been discovered, including the ``tree valuations'' in \cite{COP15}, and the valuation classes of \cite{BLN13,C14,CP18}. 
Interestingly, the former admits also a polynomial time algorithm.

Non-existence of WE has been shown for many valuation classes, mostly by constructing an ad hoc market that does not identify some particular pattern as responsible for the non-existence (e.g. \cite{GS99,LLN06,CSS05}). 
Roughgarden and Talgam-Cohen in \cite{RT15} reprove some of these results and show a systematic way of proving non-existence of WE for more general valuation and pricing classes via standard complexity assumptions. The latter paper shows the remarkable relation between computability of seemingly arbitrary problems and existence of equilibria in markets.  In fact, one of their results states that if for some class $\mathcal{V}$ of valuation functions \windet is computationally harder than finding the demand for each agent, then there exist instances in $\mathcal{V}$ with no WE.  

\paragraph{\textbf{Computation of Walrasian Equilibria.}}

On the computational side, in markets that do not guarantee existence of WE, the problem of deciding existence is \NP-hard for all the most important valuation classes. This has been established by proving that \windet for budget-additive valuations is \NP-hard via the ``knapsack'' problem in \cite{LLN06} and via the strongly \NP-hard problem ``bin packing'' in \cite{RT15}. By the fact that a WE corresponds to an optimal allocation, it is immediate that existence of WE is at least as hard as \windet. Since budget-additive functions are a subset of submodular functions, it seems that as soon as valuation functions are allowed to be more general than the class of gross substitutes, i.e. submodular, the problem is already \NP-hard. Also, for the class of single-minded agents (which is incomparable to the rest of the classes), \windet is \NP-hard \cite{CDS04}. On the positive side, Rothkopf et al.~\cite{rothkopf1998computationally} provide several classes of valuation functions where \windet can be efficiently solved. In particular, they show that when there are logarithmically many items with respect to the number of agents, \windet can be efficiently solved via dynamic programming. Murota in \cite{M96a, M96b} shows a strongly polynomial time algorithm for the problem of computing a WE in gross substitutes valuations, while Nisan and Segal in \cite{NS06} propose a way to utilize the gross substitutes properties to build a suitable linear program. Sandholm~\cite{SANDHOLM2002} provides a comparison of several different methods for \windet and experimentally evaluates them.
It is also worth mentioning the ``tollbooth'' problem on trees, defined in \cite{GHKKKM05} (see also \cite{CR08}), for which, even though WE existence is not guaranteed, finding one (if it exists) is in \p. 

\paragraph{\textbf{Relaxations/Approximations.}}
Due to \cite{GS99} and \cite{M00}, existence of WE is guaranteed only in a restrictive class of functions, namely {\em gross substitutes}. This fact has ignited a line of works that, in essence, question the initially defined WE as being the equilibrium that occurs in actual markets. These works consider relaxed or approximate versions of WE. Some of the most interesting results on such relaxations are the following: 
\begin{itemize}
	\item If only $2/3$ of the agents are required to be utility maximizers then a {\em relaxed Walrasian equilibrium} exists for single-minded agents (\cite{CDS04,CR08}).
	\item If the seller is allowed to package the items into indivisible bundles prior to sale, not all items have to be sold, and additionally only half of the optimal social welfare is required ({\em Combinatorial Walrasian equilibrium}) then such an equilibrium exists for general valuation functions and can be found in polynomial time (\cite{FGL16}).
	\item If agents exhibit {\em endowment effect}, meaning that the agents' valuations for a bundle they already possess is multiplied by a factor $a$, then for any $a \geq 2$ there exists an {\em $a$-endowed equilibrium} for the class of submodular functions (\cite{BDO18}). For stronger notions of endowment, endowed equilibria exist even for XOS functions, and additionally, bundling guarantees equilibria for general functions (\cite{EFF19}). 
\end{itemize}

Other works have also considered special classes of valuations that have as parameter the cardinality of the valuable bundles (\cite{CSS05} and \cite{CEEM04}). However these valuation functions are not identical to ours. In \cite{CSS05} the valuation function of each agent, called {\em $k$-wise dependent}, is encoded in a hypergraph whose vertices are the items and each hyperedge has a positive or negative weight that determines the additional value of the bundle in case all of its adjacent vertices are a subset of the bundle. This class of valuations is incomparable to ours by definition. The model of \cite{CEEM04} is the same as that of \cite{CSS05}, as argued in the latter. Recently, Berger et al. in \cite{BEF20} introduced a hierarchy of valuation functions similar to ours, called ``$k$-demand'' that also generalize unit-demand functions. The same definition of functions appears also in \cite{ZC19}. However, those are a special case of our $k$-demand functions (i.e. also additive), and in fact they are gross substitutes.

The paper is organized in sections so that each deals with a particular value or group of values for $k$. We study unit-demand valuations in Section \ref{sec:1-demand}, 2-demand valuations in Section \ref{sec:2-demand}, 3-demand valuations in Section \ref{sec:3-dem}, and $k$-demand valuations for constant $k$ and unbalanced markets in Section \ref{sec:const-dem}. We conclude with a discussion in Section \ref{sec:discussion}.

\section{Walrasian Equilibria and Valuation Functions}
\label{sec:preliminaries}
We consider markets with a set $N$ of $n$ agents and a set $M$ of $m$ items. Every agent $i$
has a valuation function $v_i: 2^M \to \reals_{\geq 0}$; for every subset, or bundle, 
of items $X \subseteq M$ agent $i$ has value $v_i(X)$. 
A valuation function $v_i$ is {\em monotone} if $X \subseteq Y$ implies $v_i(X) \leq v_i(Y)$, 
and it is {\em normalized} if $v_i(\emptyset)=0$.  In what follows, we assume that all the
agents have monotone and normalized valuation functions.

There are many different valuation functions studied over the years and we focus on several 
of them.\footnote{When we refer to a valuation function as {\em general} we mean that the value for any bundle does not depend on other bundles' values. It is clear that the set of general functions contains all other classes of functions.}
\begin{itemize}\setlength
	\item Unit-demand (\textsc{UD}): for agent $i$ there exist $m$ values 
	$v_{i1}, \ldots, v_{im}$	and $v_i(X) = \max_{j \in X} v_{ij}$, for every $X \subseteq M$.
	\item Additive (\textsc{AD}): for agent $i$ there exist $m$ values $v_{i1}, \ldots, v_{im}$
	and $v_i(X) = \sum_{j \in X} v_{ij}$, for every $X \subseteq M$.
	\item Budget-additive (\textsc{BA}): for every agent $i$ there exist $m+1$ values 
	$v_{i1}, \ldots, v_{im}, B_i$, such that for every $X \subseteq M$ it is $v_i(X) = \min \left\{B_i, \sum_{j \in X} v_{ij}
	\right \}$.
	\item Single-minded (\textsc{SMi}): for agent $i$ there exist a set $X_i \subseteq M$ and 
	a value	$B_i$, such that $v_i(X) = B_i$, if $X_i \subseteq X$, and $v_i(X) = 0$, otherwise.
	\item Submodular (\textsc{SubM}): for agent $i$ and every two sets of items $X$ and $Y$
	it holds $v_i(X) + v_i(Y) \geq v_i(X \cup Y) + v_i(X \cap Y)$.
	\item Fractionally subadditive (\textsc{XOS}): for every agent there exist vectors 
	$v_{i1}, \ldots v_{ik} \in \reals^m$ and 
	$v_i(X) = \max_{j \in [k]} \sum_{\ell \in X}v_{ij}(\ell)$, for every $X \subseteq M$.
	\item Subadditive (\textsc{SubA}): for agent $i$ and every two sets of items $X$ and $Y$ 	
	it holds $v_i(X) + v_i(Y) \geq v_i(X \cup Y)$.
\end{itemize}

We will focus on constrained versions of the aforementioned valuation functions, where
the cardinality of the sets an agent has value for is bounded by $k$. $k$-demand valuations naturally generalize unit-demand valuations, but, at the same time, 
they keep the structure of more complex valuation functions.

\begin{definition}[$k$-demand valuation]
	\label{def:demand-util}
	A valuation function $v:2^m \to \reals_{\geq 0}$ is $k$-demand if for every bundle 
	$X \subseteq M$ it holds that
	\begin{align*}
		v(X) = \max_{\substack{X' \subseteq X \\ |X'| \leq k}} v(X').
	\end{align*}
\end{definition}

A very important remark is that when $k$ is constant the WE related problems have succinct representation, namely polynomial in the number of agents and items, i.e. $ \Theta \left( n \cdot m^k \cdot \log V \right) $, where $V$ is the maximum valuation among all bundles and among all agents. This makes our setting computationally interesting and also removes the need for access to some {\em value oracle} or {\em demand oracle}: the former takes as input a bundle and returns its value, and the latter, for some indicated agent, takes a pricing as input and outputs the most preferable bundles for the agent. Having such oracles when $k$ is constant is redundant since there are only $\sum_{j=1}^k \binom{m}{j} \in \Theta(m^k)$ many $j$-subsets of $M$, $j \leq k$, and an agent just needs to declare a value for each; then the algorithm with this input can compute in polynomial time the value of the agent for any bundle. Also, a demand oracle is not needed since, for a given pricing, one can compute efficiently the prices of all $\sum_{j=1}^k \binom{m}{j}$ bundles (these are the only ones that can maximize the utility of an agent; by considering a bundle $Y$ with more than $k$ items, its value will correspond to a bundle $X$ with $k$ items, but $p(Y) \geq p(X)$), and then (efficiently) search through them to find which ones yield the maximum utility to the agent. In contrast, a great line of works has studied the complexity of the WE-related problems, provided that value oracles and demand oracles are available (e.g. \cite{BM97,NS06,GS99,dVSV07,LW18}).

An {\em allocation} $S = (S_0, S_1, \ldots, S_n)$ is a partition of $M$ to $n+1$
disjoint bundles, where agent $i \in [n]$ gets
bundle $S_i$. Items in $S_0$ are not allocated to any agent.
The {\em social welfare} of allocation $S$ is defined as $SW(S) = \sum_{i \in [n]}
v_i(S_i)$. An allocation $S$ is {\em optimal} if it maximizes the social welfare, i.e.,
$SW(S) \geq SW(S')$, for every possible allocation $S'$.
A {\em pricing} $p = (p_1, \ldots, p_m)$ defines a price for every item,
where $p_j \geq 0$ is the price of item $j$. For $X \subseteq M$, we denote 
$p(X) = \sum_{j \in X} p_j$. Given an allocation $S$ and a pricing $p$, the
{\em utility} of agent $i$ is 
\begin{align*}
	u_i(S,p) := v_i(S_i) - p(S_i).
\end{align*}


The {\em demand correspondence} of agent $i$ with valuation $v_i$ under pricing $p$,
denoted $D(v_i, p)$, is the set of items that maximize the utility of the agent; formally
$D(v_i,p) := \{S \subseteq M: u_i(S,p) \geq  u_i(T,p) \text{ for all}~T\subseteq M\}$. Any element of $D(v_i,p)$ is called {\em demand set} of agent $i$.

\begin{definition}[Gross substitutes (\textsc{GS})\cite{KC82}]
	A valuation function satisfies the gross substitutes property when for any price vectors $p \in \reals^m$ and $S \in D(v,p)$, if $p'$ is a price vector $p \leq p'$ (meaning that for all $l \in S$, $p_l \leq p'_l$), then there is a set $S' \in D(v,p')$ such that $S \cap \{ j; p_{j} = p'_{j} \} \subseteq S'$.
\end{definition}

Intuitively, a valuation is gross substitute if after the increase of the prices of some items in some
demand set $S$ of an agent, the agent still has a demand set $S'$ that contains the items with unchanged prices.

It is known that $\textsc{UD} \subset \textsc{BA} \subset \textsc{SubM}$, that 
$\textsc{AD} \subset \textsc{GS} \subset \textsc{SubM}$, and finally that 
$\textsc{SubM} \subset \textsc{XOS} \subset \textsc{SubA}$. Furthermore, 
\textsc{SMi} valuation functions are not contained in any of these valuation classes.

\begin{definition}[Walrasian Equilibrium]
	\label{def:WE}
	An allocation $S=(S_0,S_1,\ldots, S_n)$ and a pricing $p = (p_1, \ldots, p_m)$ form a
	{\em Walrasian equilibrium} (WE), if the following two conditions hold.
	\begin{enumerate}
		\label{we:price}
		\item  For every agent $i$ and any bundle $X \subseteq M$ it holds that 
		$v_i(S_i) - p(S_i) \geq v_i(X) - p(X)$.
		\item For every item $j \in S_0$ it holds that $p_j=0$.
	\end{enumerate}
\end{definition}

\problemdef
{Walrasian}
{A market with $n$ agents and $m$ items, and a valuation function for each agent.}
{Decide whether the market possesses a Walrasian equilibrium, and if it does, compute one.}

The {\em First Welfare Theorem} states that for any Walrasian equilibrium $(S,p)$, 
partition $S$ corresponds to an optimal allocation \cite{L17}. Hence, 
\textsc{walrasian} can be decomposed into the following two problems.

\problemdef
{WinnerDetermination}
{A market with $n$ agents and $m$ items, and a valuation function for each agent.}
{Find an optimal allocation $S^*$ for the items.}
\
\problemdef
{WalrasianPricing}
{A market with $n$ agents and $m$ items, a valuation function for each agent, and an 
	optimal allocation $S^*$.}
{Find a pricing vector $p$ such that $(S^*,p)$ is a Walrasian equilibrium, or decide that 
	there is no Walrasian equilibrium for the instance.}

This decomposition highlights that a WE exists if and only if there exists a
pricing vector $p$ that satisfies the conditions of Definition~\ref{def:WE} for any optimal 
allocation $S^* = (S_0^*,S_1^*, \ldots, S_n^*)$.

For $k$-demand valuation functions, the conditions of Definiton~\ref{def:WE} (and therefore a solution to \walpri) correspond to
the solution of the following linear system of $m$ variables and $n \cdot \sum_{j=1}^{k} \binom{m}{j}  + m$ equality/inequality constraints, where each constraint has at most $2k$ variables. 
\begin{align}\label{eq:LP}
	v_{i}(S_i^*) - p(S_i^*) \geq v_{i}(X) - p(X), \quad &\forall X \subseteq M, \text{ where } |X| \leq k, \forall i \in N  \nonumber \\
	p_{j} \geq 0, \qquad &\forall j \notin S_0^* \nonumber  \\
	p_{j} = 0, \qquad &\forall j \in S_0^*.  
\end{align}

Note that when $k$ is a constant, as mentioned earlier, the above constraints are $n \cdot \sum_{j=1}^{k} \binom{m}{j}  + m$ which is at most linear in $n$ and polynomial in $m$, since
\begin{align*}
	\sum_{j=1}^{k} \binom{m}{j}  
	&\leq \sum_{j=1}^{k} \frac{m^{j}}{j!} \leq  \sum_{j=1}^{k} \frac{k^{j}}{j!} \cdot \left(\frac{m}{k}\right)^{j}  
	\leq e^{k} \cdot  \sum_{j=1}^{k} \left(\frac{m}{k}\right)^{j} 
	\leq  e^{k} \cdot \left(\frac{m}{k}\right)^{k} .
\end{align*}

We conclude that, for constant $k$, a solution to linear system \eqref{eq:LP} (and thus, \walpri) can be found in time polynomial in $n$ and $m$ by formulating it as an LP with objective function set to a constant. So, the problem of deciding the existence of WE and the problem of computing one (if it exists) essentially reduce to finding an optimal allocation $S^*$, i.e. \windet. In Sections \ref{sec:2-demand}, \ref{sec:3-dem}, \ref{sec:const-dem} we exploit the aforementioned fact and only investigate the complexity of \windet.

\section{Unit-demand Valuation Functions}\label{sec:1-demand}

The simplest case of markets is when the agents have unit-demand valuation functions. The existence of WE in this class of markets was 
shown in the seminal paper of Demange, Gale, and Sotomayor \cite{DGS86} via an algorithm that resembles the t\^{a}tonnement process. This algorithm is 
pseudopolynomial in general, and polynomial when the values of the agents are bounded by
some polynomial. In \cite{L17} an algorithm (Algorithm 1) is presented and it is shown that a modification of it finds a WE in time $O(m^2 n + m^4 \log V)$, where $V$ is the maximum valuation of any item across all agents.

In this section we show that \textsc{walrasian} in these markets is in quasi-\NC. 
The complexity class quasi-\NC is defined as 
$\text{quasi-}\NC=\Cup_{k\geq 0} \text{quasi-}\NC^k$, where quasi-$\NC^k$ is the class 
of problems having uniform circuits of quasi-polynomial size, $n^{\log^{O(1)}n}$, and 
polylogarithmic depth $O(\log^k n)$ \cite{Bar92}.  Here  ``uniform'' means that the circuit
can be generated in polylogarithmic space. Put differently, quasi-\NC contains problems that
can be solved in polylogarithmic parallel time using quasi-polynomially many processors with 
shared memory.

In this class of markets \windet can be reduced to a maximum weight matching on a complete
bipartite graph. On the left side of the graph there exist $n$ nodes corresponding to the
agents, on the right side there are $m$ nodes corresponding to the items and the weight of the
edge $(i,j)$ equals to the value of agent $i$ for item $j$.
The recent breakthrough of Fenner, Gurjar, and Thierauf \cite{fenner2019bipartite} states
that the maximum weight {\em perfect} matching in bipartite graphs is in quasi-\NC when the 
edge-weights are bounded by some polynomial; later Svensson and Tarnawski
\cite{svensson2017matching} extended this result for general graphs.  Thus, if we augment 
the bipartite graph that corresponds to the market by adding dummy items with zero value
for every agent, or dummy agents with zero value for every item, we can guarantee that it
contains a perfect matching without changing any optimal allocation. Then, we can use the 
algorithm of \cite{fenner2019bipartite} and compute an optimal allocation in polylogarithmic 
time.

Given an optimal allocation, \walpri for these markets has a special structure. It is a linear
feasibility problem with polynomially many inequalities and at most  two variables per
inequality. For this special type of feasibility systems there exists a quasi-NC algorithm
\cite{lueker1986linear}.

\begin{theorem}
	\label{thm:qNC-algo}
	\textsc{walrasian} in unit-demand markets with polynomial valuations is in quasi-\NC.
\end{theorem}
\begin{proof}
	When shared memory is available, as in quasi-\NC, we can solve \windet in polylogarithmic 
	parallel time via the algorithm of \cite{fenner2019bipartite} and store it in the shared
	memory. Then, the processors will read the solution, build the linear system for \walpri 
	and solve it in polylogarithmic time via the  the algorithm of \cite{lueker1986linear} on the
	shared memory. Hence, the composition of the two algorithms can be done in polylogarithmic 
	time using quasi-polynomially many processors.
\end{proof}

Our result suggests a parallel algorithm that needs $O(\log^3(n))$ time which is significantly faster than any serial algorithm. On the other hand though, it requires $n^{\log(n)}$ processors in the worst case.	
We observe that this is the current best possible result, since any improvement would imply
better parallel algorithms for other important problems like maximum weight matching. We have to state though that it is open whether 
either of maximum weight matching or feasibility of a system with linear inequalities and two variables per constraint are in \NC. On the other hand, it is known that the maximum
weight problem in graphs with polynomial weights is in pseudo-deterministic \RNC 
\cite{AV20,GG17}. Hence, a first improvement would be to
place \walpri in pseudo-deterministic \RNC.

\section{2-demand Valuation Functions}\label{sec:2-demand}
In this section we resolve the complexity of deciding existence of WE for 2-demand valuation 
functions. As an example, consider the case where the football teams need to have at least 2 young native players in their squad. Each team knows exactly which pair of players wants and it does not want more young players due to capacity constraints. A version of 2-demand valuations, termed {\em pair-demand} valuations, was studied in \cite{RT15}, where every agent $i$ has a value $v_i(j,k)$ for every pair of 
items and the value of $i$ for a bundle $S$ is $v_i(S)=\max_{j,k \in S}v_i(j,k)$. These 
are general valuation functions that can allow complementarities.  We strengthen the results 
of \cite{RT15} and prove that \windet is \NP-hard even when the valuation functions of 
the agents are 2-demand \textsc{XOS} and every agent has positive value for at most six items. 

\begin{theorem}
	\label{thm:2card-sub}
	\windet is strongly \NP-hard even for 2-demand \textsc{XOS} functions.
\end{theorem}
\begin{proof}
	We reduce from 3-bounded 3-dimensional matching, termed \tdm. 
	The input of a \tdm instance consists of three sets $X,Y,Z$, where $|X|=|Y|=|Z|$,  and a set 
	$S$ of triplets (hyperedges) $(x,y,z)$ where $x \in X, y \in Y$, and $z \in Z$. In addition, every element of 
	$X, Y, Z$ appears in at most three triplets and every triplet shares at most one element with any other triplet. 
	The task is to decide if there is a subset of non-intersecting triplets of $S$ of cardinality $|X|$. 
	The problem is known to be \NP-complete \cite{Kann91}.
	
	For every element $x \in X$ we create an agent and for every element in $Y \cup Z$ we 
	create an item. Let $S_{x_j}$ denote the set of items that correspond to the $j$th triplet
	of $S$ that $x$ belongs to. Recall that there exist at most three such triplets. In addition, 
	since any two triplets of $S$ share at most one element, we have that $S_{x_j}$s are disjoint.
	
	For each agent $x \in X$ we consider her corresponding $j \in \{1,2,3\}$, i.e. the number of triplets in which $x$ appears (if $x$ does not appear in any triplet then the answer to \tdm is trivial). For $j' \in [j]$ we define the vector $v_{xj'} \in \reals^m$, whose $t$-th coordinate ($t \in [m] := M$) is
	\begin{align*}
		v_{xj'}(t) =
		\begin{cases}
			1 \quad &\text{, if} \quad t \in S_{x_{j'}},\\
			0  \quad &\text{, otherwise}.
		\end{cases}
	\end{align*}
	The valuation of agent $x$ for some bundle $T \subseteq M$ is $v_{x}(T) = \max_{j' \in [j]} \sum_{t \in T}v_{xj'}(t)$. This, by definition, is an \textsc{XOS} valuation function.
	
	We claim that there is an allocation with welfare $2|X|$ if and only if the \tdm instance is 
	satisfiable. Firstly, assume that indeed the \tdm instance has a solution $S'$, i.e., $S'$ 
	contains $|X|$ non intersecting triplets in $S$. Then, if the triplet $(x,y,z)$ belongs to $S'$
	we allocate the items that correspond to $y$ and $z$ to the agent that corresponds to $x$
	and the agent has value 2 for the bundle. Clearly, the allocation achieves welfare $2|X|$.
	For the other direction, assume that there is an allocation for the items with welfare $2|X|$.
	This means that every agent gets utility 2 from her allocated bundle. Then, by construction,
	each agent $x$ alongside her allocated bundle corresponds to a triplet from $S$. Observe, that
	the allocation consists of non-overlapping bundles, hence we get $|X|$ non intersecting triplets in
	$S$.
\end{proof}

Theorem~\ref{thm:2card-sub} implies that \textsc{walrasian} is \NP-hard for any class of valuation functions that
contains the class of 2-demand \textsc{XOS} valuations. 
\begin{corollary}
	\textsc{walrasian} is strongly \NP-hard even if all the agents have 2-demand \textsc{XOS} valuation functions.
\end{corollary}

\paragraph{\textbf{Closing the gap in single-minded valuations.}}
In addition to the above hardness results we study single-minded agents with 2-demand valuations and 
we show that in this case \textsc{walrasian} is easy, contrary to the case of 3-demand valuations where it is \NP-hard \cite{CDS04}. To prove this, for agents that are single-minded for
bundles of size 2, we reduce \windet to a maximum weight matching
problem over a graph $G$. Every item corresponds to a vertex of $G$. 
For every pair of items that is the most preferable by an agent we create the corresponding edge with
weight the value of the agent for the items; if there are more than one agents that want the 
same pair of items we keep only the weight for the highest valuation. Clearly, any maximum 
weight matching corresponds to an optimal allocation.

Next we show how to handle instances where every agent is either unit-demand or 
{\em multi-minded over a subset of size 2}. Recall, a unit-demand agent-might have positive value for various items. An agent $i$ is multi-minded over a subset of size 2, if there exist items $a_i$ and $b_i$ and the agent has positive values only for the following three bundles: $\{a_i\}$, $\{b_i\}$, and $\{a_i, b_i\}$. Observe the this is a strict generalisation of 2-demand single minded.
To  achieve this, we extend the construction described above as follows. For every multi-minded agent $i$ we add a new vertex $x_i$ and the edges $(x_i,a_i)$, $(x_i,b_i)$, with weights $v_i(a_i)$ and $v_i(b_i)$ respectively. 
For every unit-demand agent $i$, we add a new vertex $y_i$ and the edges $(y_i,j)$, where $j$ is a vertex that corresponds to item $j$, with weight $v_i(j)$; i.e. equal to the agent's value for item $j$.
Again, a maximum weight  matching for the constructed graph corresponds to an optimal 
allocation.

\begin{theorem}
	\textsc{walrasian} is in \p for markets where every agent is unit-demand or multi-minded over a subset of size 2.
\end{theorem}


\section{3-demand Valuation Functions}\label{sec:3-dem}
In this section we prove strong \NP-hardness for \windet for 3-demand budget-additive 
valuation functions. 
As an example for 3-demand budget additive valuation, we can think of departments within a university that want to hire staff members for their labs. The agents are the departments, the items are the staff members, and the available resources of each department's lab defines the budget. The value a department gets from a candidate equals the quantity of the resources the candidate is capable of utilizing. The department is allowed to hire at most 3 staff members, due to regulations imposed by the university. 
\begin{theorem}
	\label{thm:three-ba}
	\windet is strongly \NP-hard even when all the agents have identical 3-demand budget-additive
	valuation functions.
\end{theorem}
\begin{proof}
	We prove the theorem with a reduction from 3-partition. An instance of 3-partition consists of 
	a multiset of $3n$ positive integers $a_1, a_2, \ldots, a_{3n}$ summing up to $S$. The question 
	is whether the multiset can 	be partitioned into $n$ triplets such that the elements of each triplet 
	sum up to $B = \frac{S}{n}$. So, given an instance of 3-partition we create a \windet instance with 
	$n$ agents and $3n$ items. 
	All the agents have the same 3-demand budget-additive valuation: they have value $a_i$ for 
	item $i$ and budget $B$. 
	
	The question we would like to decide is whether there exists an allocation 
	with social welfare $n\cdot B$.
	It is not hard to see that if there is a solution to 3-partition, then there 
	exists an allocation for \windet with social welfare $n\cdot B$. On the other hand, observe that, due to 
	the budget-additive valuations, social welfare $n \cdot B$ for the instance can be achieved only when 
	there exists an allocation where every agent gets value $B$. 
	In addition, since the agents have 3-demand valuation functions it means that any allocation that maximizes
	the social welfare, without loss of generality, allocates exactly three items to every agent; otherwise some agent gets more than 3 items and value gets wasted since, by definition of 3-demand valuation, the agent will only appreciate the 3 most valuable  items. 
	Hence, if there
	exists an allocation for the constructed instance with social welfare $n \cdot B$, necessarily, every 
	agent gets exactly 3 items whose values sum up to $B$. This allocation trivially defines a solution to 
	3-partition.
\end{proof}

\begin{corollary}
	\textsc{walrasian} is strongly \NP-hard even if all the agents have  identical 3-demand budget-additive valuation functions.
\end{corollary}

\section{Constant-demand Valuation Functions}\label{sec:const-dem}
In this section we study markets where the agents have $k$-demand valuation functions,
where $k$ is constant. Our results from the previous sections imply that deciding the existence of a WE  is 
\NP-hard even when $k=2$ and the valuation functions are \textsc{XOS}. In addition, we showed that the problem is \NP-hard for $k=3$ even for budget-additive valuations \cite{CSS05}.
This means that in order to get efficient algorithms we have to further restrict our market design in markets that retain constant demand $k$, but with either reduced number of agents, or reduced number of items.
For this reason, we study {\em unbalanced} markets. A market is {\em unbalanced} if the number of available items is significantly larger than the
number of agents, formally, $m \in \omega(n)$, or the other way around, $n \in \omega(m)$. For the case where, $m = O(\log n)$ and any $k$ the dynamic programming approach of Rothkopf et al. \cite{rothkopf1998computationally} solves \windet in $O(n^3)$. Next we show a result for the case where the market is unbalanced in the opposite direction.

\begin{theorem}
	\label{thm:m>n}
	In markets with $k$-demand valuations, $n$ agents and $m$ items, where $k$ and $n$ are constant, \windet is in \p.
\end{theorem}
\begin{proof}
	We consider the unbalanced market where the number of available items $m$ is a lot greater than the number of items $k \cdot n$ to be allocated. The number $k \cdot n$ comes from the fact that in an optimum allocation, not more than $k \cdot n$ items will be appreciated by the agents (by definition of the $k$-demand valuation function). Therefore, allocating more than these items does not improve the social welfare, thus, does not yield additional WE. In this case, we can find all possible subsets of size $k \cdot n$ of items, that is, all candidate sets of items to be allocated to the agents. Formally, we consider the set $I := \{ L \subseteq M ~ | ~ |L|=k\cdot n \}$ that consists of all $(k \cdot n)$-subsets of $M$. 
	It is $|I| = \binom{m}{k\cdot n} \in O((m-k\cdot n)^{k \cdot n})$, which is a polynomial in $m$ when $k$ and $n$ are constant.
	
	Observe now that, given a subset $L$ of items with size $k\cdot n$, one can construct a 
	$k+1$-uniform hypergraph, i.e. a hypergraph all of whose hyperedges have size $k+1$, in 
	the following way. Have its vertex set be $L \cup N$, and for every $k$-subset $L_k$ of $L$ have a hyperedge $L_k \cup \{i\}$ for every $i \in N$. Also, assign to each hyperedge a weight equal to the valuation of agent $i$ for the item bundle $L_k$, namely $v_{i}(L_k)$. On this graph one can run a brute-force algorithm to find a maximum weight ($k+1$)-dimensional matching in constant time, since the graph is of constant size. Then, by repeating the same routine for all $(k\cdot n)$-subsets of $I$ in time polynomial in $m$, we pick the one that yields the maximum sum of weights in the matching. The optimal allocation of items to agents corresponds to the aforementioned optimum matching. The running time of this algorithm is $O(m^c)$ for some constant $c$, i.e. polynomial in the input size, since the input size is $\Omega\left(n \cdot \binom{m}{k} \cdot \log V \right)$ bits, where $V := \max_{\substack{i \in N\\X \subseteq M}} v_{i}(X)$; that is because every agent has to declare how much her valuation is for every $k$-subset of items.
\end{proof}

\begin{corollary}
	In markets with $k$-demand valuations, $n$ agents and $m$ items, where $k$ and $n$ are constant, \textsc{walrasian} is in \p.
\end{corollary}

\section{Discussion}\label{sec:discussion}
In this paper we study the complexity of computing Walrasian equilibria in markets with 
$k$-demand valuations.
As we show, even for $k=2$, the problem of deciding WE existence remains \NP-hard for a relatively restricted class of valuation functions, not very far outside gross substitutes, i.e. \textsc{XOS}. For $k=3$ the problem is already \NP-hard even for budget-additive valuations. Hence, we turn to the study of unbalanced markets and present a  polynomial-time algorithm for $k$-demand general valuations, where $k$ is constant.

For markets with $k=1$, known as ``matching markets'', we prove 
that the problem is in quasi-\NC. We view this as a 
very interesting result since all the known algorithms for the problem are highly sequential.
Can we design an \NC algorithm for the problem via a form of a simultaneous auction? This
would be remarkable since it would imply that bipartite weighted matching is in \NC. 
For $k=2$ we show that \windet is intractable even for \textsc{XOS} functions, and for $k=3$ the hardness remains for an even stricter class, namely budget-additive functions.
In order to completely resolve the complexity of 2-demand valuations, it remains to solve 
\windet for 2-demand budget-additive and 2-demand submodular valuations. Is the problem \NP-hard, or is there a 
polynomial time algorithm for it? Answering this question would provide a complete dichotomy 
for the complexity of the problems \windet and also \textsc{Walrasian}. 
For unbalanced markets with constant $k$, we covered the cases $m \in O(\log n)$ and $n \in \Theta(1)$. Are there efficient algorithms for any $m \in \omega(\log n)$ or $n\in \omega(1)$?
Another very intriguing direction is to study approximate
Walrasian equilibria. The recent results of Babaioff, Dobzinski, and Oren
\cite{BDO18} and of Ezra, Feldman, and Friedler \cite{EFF19}
propose some excellent notions of approximation. Can we get better results if we assume 
$k$-demand valuations?
	



\section*{Acknowledgements}
The work of the second author was supported by the Alexander von Humboldt Foundation with funds from the German Federal Ministry of Education and Research (BMBF). The work of the third author was partially supported by the NeST initiative of the School of EEE and CS at the University of Liverpool, and by the EPSRC grant EP/P02002X/1.



\bibliographystyle{plain}
\bibliography{references}


\end{document}